%% file: main.tex
\documentclass[11pt]{article}
\usepackage{booktabs} % For formal tables
\usepackage[numbers]{natbib}
\usepackage{amsmath,amssymb,amsthm}
\usepackage[ruled]{algorithm2e} % For algorithms
\usepackage{subcaption}
\usepackage{epigraph}
\usepackage{graphicx}
\usepackage{thmtools}
\usepackage{thm-restate}
\usepackage{hyperref}
\usepackage{cleveref}

\setcitestyle{acmnumeric}

\setcitestyle{acmnumeric}

\newtheorem{theorem}{Theorem}
\newtheorem{observation}[theorem]{Observation}
\newtheorem{claim}[theorem]{Claim}
\newtheorem{lemma}[theorem]{Lemma}
\newtheorem{corollary}[theorem]{Corollary}
\newtheorem{proposition}[theorem]{Proposition}

\newtheorem{definition}[theorem]{Definition}

\newcommand{\rev}[1]{\mathrm{Rev} ( #1 )}
\newcommand{\brev}[1]{\mathrm{BRev} ( #1 )}

\newcommand{\arev}[1]{\mathrm{ARev} (#1)}

\newcommand{\realD}{\mathcal{D}}
\newcommand{\gap}{\mathrm{gap}}
\newcommand{\sgap}{\mathrm{sgap}}

\newcommand{\SupGap}[1]{\mathrm{SupGap} ( #1 )}
\newcommand{\MenuGap}[1]{\mathrm{MenuGap} ( #1 )}
\newcommand{\ScaGap}[1]{\mathrm{AlignGap} ( #1 )}

\newcommand{\LagRel}[1]{\mathrm{LagRel} (#1)}
\newcommand{\LagRelOne}[1]{\mathrm{LagRel}_1 (#1)}
\newcommand{\LagRelTwo}[1]{\mathrm{LagRel}_2 (#1)}
\newcommand{\norm}[1]{\left\lVert#1\right\rVert}

\begin{document}

% Title. Note the optional short title for running heads. In the interest of anonymization, please do not include any acknowledgements.
\title{On Infinite Separations Between Simple and Optimal Mechanisms}

% Anonymized submission.
\author{
C. Alexandros Psomas
\thanks{Department of Computer Science, Purdue University, apsomas@cs.purdue.edu.} \and 
Ariel Schvartzman
\thanks{Center for Discrete Mathematics and Theoretical Computer Science (DIMACS), Rutgers University, as3569@dimacs.rutgers.edu. Supported NSF CCF-1445755.} \and
S. Matthew Weinberg
\thanks{Department of Computer Science, Princeton University, smweinberg@princeton.edu. Supported by NSF CCF-1717899}
}

\maketitle

\begin{abstract}
We consider a revenue-maximizing seller with $k$ heterogeneous items for sale to a single additive buyer, whose values are drawn from a known, possibly correlated prior $\realD$. It is known that there exist priors $\realD$ such that simple mechanisms --- those with bounded menu complexity --- extract an arbitrarily small fraction of the optimal revenue~(\citet{BriestCKW15,hart2019selling}). This paper considers the opposite direction: given a correlated distribution $\realD$ witnessing an infinite separation between simple and optimal mechanisms, what can be said about $\realD$?

\citet{hart2019selling} provides a framework for constructing such $\realD$: it takes as input a sequence of $k$-dimensional vectors satisfying some geometric property, and produces a $\realD$ witnessing an infinite gap. Our first main result establishes that this framework is without loss: \emph{every} $\realD$ witnessing an infinite separation could have resulted from this framework. An earlier version of their work provided a more streamlined framework~\cite{HNv1}. Our second main result establishes that this restrictive framework is \emph{not} tight. That is, we provide an instance $\realD$ witnessing an infinite gap, but which provably could not have resulted from the restrictive framework. 

As a corollary, we discover a new kind of mechanism which can witness these infinite separations on instances where the previous ``aligned'' mechanisms do not.
\end{abstract}

\input{intro}
\input{notation}

\input{converse}
\input{separation}
\input{conclusion}

\bibliographystyle{ACM-Reference-Format}
\bibliography{refs}

\newpage

\appendix
\input{appendix}

\end{document}

%% file: intro.tex
\section{Introduction}
\label{sec:intro}

%\epigraph{Beware that, when fighting monsters, you yourself do not become a monster... for when you gaze long into the abyss. The abyss gazes also into you.}{\textit{Friedrich Nietzsche}}

Consider a revenue-maximizing seller with $k$ items for sale to a single additive buyer, whose values for the $k$ items are drawn from a known distribution $\realD$. When $k=1$, Myerson's seminal work provides a closed-form solution to the revenue-optimal mechanism, and it has a particularly simple form: simply post a price $p:=\arg\max_{p}\{ p \cdot \Pr_{v \leftarrow \realD}[v \geq p]\}$, and let the buyer purchase the item if they please~(\citet{Myerson81}). For $k > 1$, however, this \emph{multi-dimensional mechanism design} problem remains an active research agenda forty years later.

While simple, constant-factor approximations are known in quite general settings when $\realD$ is a product distribution~(\citet{ChawlaHK07, ChawlaHMS10, ChawlaMS15, KleinbergW12, LiY13, BabaioffILW14, Yao15, RubinsteinW15, CaiDW16, ChawlaM16, CaiZ17}), there may be an \emph{infinite} gap between the revenue-optimal auction and any simple counterpart when values are correlated~(\citet{BriestCKW15,hart2019selling}). More specifically: one simple way to sell $k$ items is to treat the grand bundle of all items as if it were a single item, and sell it using Myerson's optimal auction (which sets price $\arg\max_p\{ p \cdot \Pr_{\vec{v} \leftarrow \realD}[\sum_i v_i \geq p] \}$). Letting $\brev{\realD}$ denote the revenue achieved by this simple scheme,~\citet{hart2019selling} further show a connection between $\brev{\realD}$ and \emph{any} simple mechanism through the lens of menu complexity: any mechanism with menu complexity at most $C$ generates expected revenue at most $C \cdot \brev{\realD}$. 

These works establishes a strong separation between simple and optimal auctions: \emph{even when $k=2$}, there exist distributions $\realD$ such that $\rev{\realD} = \infty$ (the optimal revenue) while $\brev{\realD}=1$. The fact that $\rev{\realD} = \infty$ does not on its own suggest that $\realD$ must be ``weird'' (the one-dimensional distribution with CDF $1-1/\sqrt{x}$ has this property). The ``weird'' property is that $\rev{\realD}/\brev{\realD} = \infty$, which can never occur for a one-dimensional distribution. 

\citet{BriestCKW15} and~\citet{hart2019selling} establish sufficient conditions for a distribution $\realD$ to satisfy $\rev{\realD}/\brev{\realD} = \infty$. Simply put, the goal of this paper is to study \emph{necessary} conditions for a distribution to satisfy $\rev{\realD}/\brev{\realD}= \infty$. We provide two main results. The first establishes that the sufficient condition presented in~\cite{hart2019selling} \emph{is in fact necessary} for $\rev{\realD}/\brev{\realD} = \infty$ (Theorem~\ref{thm:main}). The second establishes that the sufficient condition used in an earlier version of that work~\cite{HNv1} and~\cite{BriestCKW15} is \emph{not necessary} (Theorem~\ref{thm:main2}). In establishing Theorem~\ref{thm:main2}, we also construct a distribution $\realD$ such that $\rev{\realD}/\brev{\realD} = \infty$ yet the mechanism witnessing this provably falls outside the scope of any previous constructions (Corollary~\ref{cor:main}). 

Proper context and formal statements of our results require precise definitions, which we provide in Section~\ref{sec:prelims} immediately below. Section~\ref{sec:results} provides formal statements of our results, along with context alongside related work. Subsequent sections provide proofs.

%% file: notation.tex
\section{Preliminaries}\label{sec:prelims}

We consider an auction design setting with a single buyer, single seller, and $k$ heterogeneous items. Note that our positive results hold for arbitrary $k$, while our constructions use only $k=2$ (and $k=1$ is not possible). We use $\realD$ to denote a distribution over $\mathbb{R}_{\geq 0}^k$, the (possibly correlated) distribution over the buyer's values for the $k$ items. The buyer is additive, meaning that their value for a set of items $S$ is equal to $\sum_{i \in S} v_i$. We use $\rev{\realD}$ to denote the optimal expected revenue achievable by any incentive-compatible mechanism (formally, the supremum of expected revenues, or $\infty$ if the supremum is undefined), and let $\brev{\realD}$ denote the revenue achieved by selling the grand bundle as a single item using Myerson's auction.\footnote{We briefly remind the reader that $\brev{\realD}$ serves as a proxy for the achievable revenue by any simple mechanism, especially when focusing on the gap between infinite and finite. For example, the revenue achieved by selling the items separately is at most $k\brev{\realD}$, the revenue achieved by any deterministic mechanism is at most $2^k\brev{\realD}$, and, more generally, the revenue achieved by any mechanism which offers at most $m$ distinct options is at most $m\brev{\realD}$~\cite{hart2019selling}.} Finally, a \emph{mechanism} $M$ is a set $\{(\vec{q}_i, p_i)\}_i$, where each $\vec{q}_i \in [0,1]^k$ denotes a vector of probabilities, and $p_i \in \mathbb{R}$ denotes a price. When the buyer's valuation is $\vec{v}$, they pay the auctioneer $p_{i(\vec{v})}$, where $i(\vec{v}):=\arg\max_i\{\vec{v} \cdot \vec{q}_i - p_i\}$.\footnote{How ties are broken is irrelevant to our results --- all results hold for arbitrary tie-breaking. Also, all mechanisms include an all-zero pair $ \vec{q}_0 = (0, \dots, 0)$, $p_0 = 0$ to ensure individual rationality.} $\rev{\realD,M}$ denotes the expected revenue achieved by a particular mechanism $M$ on distribution $\realD$. We will also use the shorthand $\vec{q}^M(\vec{v})$ to denote the allocation vector purchased by a vector $\vec{v}$, and $p^M(\vec{v})$ to denote the price paid (we may drop the superscript of $M$ if the mechanism is clear from context).

\subsection*{Brief Overview of \citet{hart2019selling}}

Below, we formally define two geometric properties of sequences of points, which are the focus of this paper. Many of the ideas below appear in both~\citet{BriestCKW15} and~\citet{hart2019selling}, but we will use the formal definitions from~\citet{hart2019selling} (and an earlier published version~\citet{HNv1}). Below, morally the vectors $\vec{x}_i$ correspond to possible (scaled) valuation vectors, and the vectors $\vec{q}_i$ correspond to possible vectors of allocation probabilities. 

\begin{definition}
\label{def:sergiu}
Let $X=(\vec{x}_i)_{i=1}^N$ be an ordered sequence of $N$ points ($N$ may be finite, or equal to $+\infty$), with each $\vec{x}_i \in \mathbb{R}_{\geq 0}^k$. Let $Q=(\vec{q}_i)_{i=0}^N$ be another ordered sequence of $N$ points, with each $\vec{q}_i \in [0,1]^k$, and starting with $\vec{q}_0 = (0,...,0)$. Define the following: 
\begin{align*}
\gap_i^{X,Q} := \min_{0 \leq j < i} (\vec{q}_i - \vec{q}_j) \cdot \vec{x}_i && \MenuGap{X,Q} := \sum_{i=1}^N \frac{\gap_i^{X,Q}}{||\vec{x}_i||_1}.
\end{align*}

We will also slightly abuse notation and define $\MenuGap{X}:=\sup_Q \{\MenuGap{X,Q}\}$.\footnote{Observe that $\gap_i^{X,Q}$ and $\MenuGap{X,Q}$ might be negative. Any claims made throughout this paper regarding $\MenuGap{X,Q}$ are vacuously true when $\MenuGap{X,Q} < 0$ (e.g. Theorem~\ref{thm:newHN}). We allow $\gap_i^{X,Q}$ to be negative to match the definition of~\cite{hart2019selling} verbatim (although our work will also show that this peculiarity of their definition is not significant).}
\end{definition}

Intuitively, $\MenuGap{X}$ is some (complicated) measure of how distinct the angles of points in $X$ are. To get intuition for this, one might try to write a short proof that when $k=1$, $\MenuGap{X} = 1$ for all $X$ (or that $\MenuGap{X}=1$ whenever all $\vec{x}_i \in X$ are parallel). We provide such a proof in Appendix~\ref{app:single}.

Still, $\MenuGap{X,Q}$ is just some geometric measure with no obvious intuition for why this quantity should be of interest to auction designers. However, one key result of~\cite{hart2019selling} shows that this quantity has connections to simplicity vs.~optimality gaps. Specifically, they show:

\begin{theorem}[\cite{hart2019selling}, Proposition~7.1] 
\label{thm:newHN}
For every pair of sequences $X=(\vec{x}_i)_{i=1}^N, Q=(\vec{q}_i)_{i=0}^N$ starting with $\vec{q}_0 =  (0,...,0)$, and all $\varepsilon > 0$, there exists a distribution $\realD$ and mechanism $M$ such that:

\[ \frac{\rev{\realD,M}}{\brev{\realD}} \geq (1-\varepsilon)\cdot  \MenuGap{X,Q}.\]

Moreover, for all $i \in [N]$, the support of $\realD$ contains a single point of the form $c_i \vec{x}_i$, for some $c_i \in \mathbb{R}_+$ (and no other points). Additionally, $\vec{q}^M(c_i \vec{x}_i) = \vec{q}_i$.
\end{theorem} 

The ``Moreover,...'' portion of Theorem~\ref{thm:newHN} gives some insight to their construction. Further insight can be deduced by observing that the constraint ``$\gap_i^{X,Q} \leq \vec{x}_i \cdot (\vec{q}_i - \vec{q}_j)$'' looks similar (but far from identical) to an incentive compatibility constraint involving a valuation vector $\vec{x}_i$ and two allocation vectors $\vec{q}_i,\vec{q}_j$. We refer the reader to~\cite{hart2019selling} for further details and intuition for this connection. Theorem~\ref{thm:newHN} gives a framework for proving simplicity vs.~optimality gaps, but of course leaves open the question of actually finding a pair of sequences $X, Q$. They approach this through the following observation:

\begin{definition}
\label{def:noam}
Given a sequence of points $X = (\vec{x}_i)^N_{i=0} \in [0,1]^k$ starting with $\vec{x}_0 = (0,...,0)$, define $\SupGap{X}$ (read ``support gap of $X$") as follows: 
\begin{align*}
\SupGap{X} := \MenuGap{X,X}.
\end{align*}
\end{definition}

\begin{observation}\label{obs:trivial} For all $X$, $\MenuGap{X} \geq \SupGap{X}$.
\end{observation}

Finally,~\cite{HNv1} propose an explicit construction of a sequence $X$ with infinite SupGap.

\begin{theorem}[\cite{HNv1}]\label{thm:HNgap} There exists an infinite sequence $X$ of points in $[0,1]^2$ such that:
$$\SupGap{X} = \infty.$$
\end{theorem}

Theorems~\ref{thm:newHN} and~\ref{thm:HNgap} together yield a two-dimensional distribution with $\rev{\realD}/\brev{\realD}=\infty$. It is also worth noting that all prior constructions follow this approach as well. For example,~\cite{BriestCKW15} provides an infinite sequence $X$ of points in $[0,1]^3$ such that $\SupGap{X} = \infty$.

\section{Our Results}\label{sec:results}
Independent of Observation~\ref{obs:trivial} and Theorem~\ref{thm:HNgap}, Theorem~\ref{thm:newHN} \emph{alone} provides a framework for constructing distributions $\realD$ so that $\rev{\realD}/\brev{\realD}$ is large: find sequences $X$ so that $\MenuGap{X}$ is large. Our goal is to understand to what extent this framework is \emph{complete} for constructing such instances. Our first main result establishes that \emph{any distribution with $\rev{\realD}/\brev{\realD} = \infty$ could have resulted from the framework induced by Theorem~\ref{thm:newHN}}. Specifically:

\begin{theorem}\label{thm:main}
For any distribution $\realD$ over $k$ items, there exists a sequence of $N$ points ($N$ can be finite, or equal to $+\infty$) $X=(\vec{x}_i)^N_{i=1}$, with each $\vec{x}_i \in \text{supp}(\realD)$, such that 
\[ \MenuGap{X} \geq \frac{\rev{\realD}}{9\brev{\realD}}. \]

In particular, if $\rev{\realD}/\brev{\realD} = \infty$, then $\MenuGap{X} = \infty$ as well.
\end{theorem}

A complete proof of Theorem~\ref{thm:main} appears in Section~\ref{sec:converse}. Observe also that because $\MenuGap{X}$ is monotone (in the sense that adding points to $X$, anywhere, cannot possibly decrease $\MenuGap{X}$), the fact that $X$ is a subset of the support of $\realD$ (rather than the entire support) is immaterial.\footnote{Of course, if the support of $\realD$ is uncountable, then clearly the entire support of $\realD$ cannot be included in $X$.} Put another way, the important aspect in constructing $X$ is how elements in the support of $\realD$ are ordered, rather than which points are included.

Observation~\ref{obs:trivial} further provides a framework to construct sequences so that $\MenuGap{X}$ is large: construct sequences $X$ so that $\SupGap{X}$ is large. One may then wonder if $\SupGap{X}$ and $\MenuGap{X}$ are approximately related, for all $X$. For this specific question, the answer is trivially no, due to incompatibility with scaling (multiplying every point in $X$ by $1/2$ will decrease $\SupGap{X}$ by a factor of $2$, but not $\MenuGap{X}$). Therefore, not much insight is gained by studying this precise question.

Instead, we observe that the interesting aspect of constructions resulting through $\SupGap{X}$ is that $\vec{x}_i$ and $\vec{q}_i$ are \emph{aligned} (that is, $\vec{x}_i = c_i \cdot \vec{q}_i$ for some $c_i \in \mathbb{R}_{\geq 0}$). Specifically, even if $\vec{q}_i = \vec{x}_i$, this equality is not maintained through the construction of Theorem~\ref{thm:newHN}. However, if $\vec{q}_i$ and $\vec{x}_i$ are aligned, this alignment property \emph{is} maintained by the construction. We therefore propose the following definition, which captures the maximum value achievable by $\MenuGap{X,Q}$ when $X,Q$ are aligned.

\begin{definition}\label{def:align} Let $X = (\vec{x}_i)_{i=1}^N$ be an ordered sequence of $N$ points in $[0,1]^k$ ($N$ may be finite, or equal to $+\infty$). Let also $C=(c_i)_{i=0}^N$ be an ordered sequence of numbers, with each $c_i \in [0,1/||\vec{x}_i||_\infty]$, starting with $c_0=0$. Define $\sgap^{\cdot,\cdot}$ (read "scalar gap") :$$\sgap^{X,C}_i:= \min_{j < i} \vec{x}_i \cdot (c_i \vec{x}_i - c_j \vec{x}_j),\qquad  \text{ and } \qquad \ScaGap{X,C}:= \sum_{i=1}^N \frac{\max\{0,\sgap^{X,C}_i\}}{||\vec{x}_i||_1}.$$

We will also slightly abuse notation and denote by $\ScaGap{X}:=\sup_C\{\ScaGap{X,C}\}$.
\end{definition}

Recall that we have chosen to let $c_i$ range in $[0,1/||\vec{x}_i||_\infty]$ (rather than be fixed at $1$, or $1/||\vec{x}_i||_\infty$) to give potential constructions flexibility in scaling $\vec{q}_i$. Additionally, by ensuring that the contribution of each $\sgap_i^{X,C}$ is non-negative, we give potential constructions flexibility to ignore points in the sequence. That is, any construction using MenuGap directly can always set $\vec{q}_i:=\arg\max_{j < i}\{\vec{q}_j\cdot \vec{x}_i\}$, which effectively just drops $\vec{x}_i$ from the sequence. Counting $\max\{0,\sgap_i^{X,C}\}$ towards the objective (rather than just $\sgap_i^{X,C}$) gives constructions that arise through AlignGap the same flexibility. 

\begin{lemma}\label{lem:align} For all $X$, $\ScaGap{X} \leq \MenuGap{X}$.
\end{lemma}

The proof of Lemma~\ref{lem:align} is in Appendix~\ref{app:omitted}. Lemma~\ref{lem:align} induces a framework to design sequences with large MenuGap: design sequences with large AlignGap. Our second main result establishes that this framework \emph{is not} without loss of generality, even for $k=2$. Specifically:

\begin{theorem}\label{thm:main2}
There exist sequences $X=(\vec{x}_i)_{i=1}^\infty \in [0,1]^2$ such that:

$$\ScaGap{X} \leq 6 \qquad \text{ but } \qquad \MenuGap{X} = \infty.$$
\end{theorem}

A complete proof of Theorem~\ref{thm:main2} appears in Section~\ref{sec:sep}. By the discussion following Definition~\ref{def:align}, the source of this gap is entirely due to the requirement that the sequence $Q$ be aligned with $X$ (and is not due to inability to scale, or inability to ignore difficult points in $X$). We make this crisp with the following definition and corollary, which construct a novel distribution witnessing $\rev{D}/\brev{D} = \infty$ that is provably distinct from all previous approaches.

\begin{definition} For a distribution $\realD$ and mechanism $M$, define the \emph{Aligned Revenue} of $M$ on $\realD$:
$$\arev{\realD,M}:=\mathbb{E}_{\vec{v} \leftarrow \realD}[p^M(\vec{v}) \cdot I(\text{$\vec{v}$ is parallel to $\vec{q}^M(\vec{v})$})], \qquad \quad \arev{\realD}:=\sup_M\{\arev{\realD,M}\}.$$
\end{definition}

\begin{corollary}\label{cor:main} There exist distributions $\realD$ over two items such that $\rev{\realD}/\arev{\realD} = \infty$.
\end{corollary}

A proof of Corollary~\ref{cor:main} appears in Appendix~\ref{app:main}. It is worth noting that \emph{all} previous constructions establishing $\rev{\realD}/\brev{\realD} = \infty$ proceeded by producing an $X$ such that $\SupGap{X} = \infty$. Indeed, the~\cite{BriestCKW15} construction provides such an $X$ when $k=3$, the~\cite{HNv1} construction provides an $X$ when $k=2$, and the~\cite{PsomasSW19} construction adapts parameters in that of~\cite{HNv1}. By Theorem~\ref{thm:newHN}, this implies not only that $\rev{\realD}/\brev{\realD} = \infty$, but also that $\arev{\realD}/\brev{\realD} = \infty$. Corollary~\ref{cor:main} establishes the existence of a fundamentally different construction,\footnote{On a technical level, our construction certainly borrows several ideas from previous ones, however.} as our $\realD$ has a finite ratio between $\arev{\realD}/\brev{\realD}$, yet still maintains an infinite ratio between $\rev{\realD}/\brev{\realD}$. 

\subsection*{Additional Related Work} 
\label{sec:rel}
We've already discussed the most related work to ours, which is that of~\citet{hart2019selling, BriestCKW15}. There is also a large body of work studying \emph{product distributions} specifically, and establishes that simple mechanisms can achieve constant factor approximations in quite general settings~\citet{ChawlaHK07, ChawlaHMS10, ChawlaMS15, KleinbergW12, LiY13, BabaioffILW14, Yao15, RubinsteinW15, CaiDW16, ChawlaM16, CaiZ17}). Recent works have made progress in obtaining arbitrary approximations~(\citet{BabaioffGN17, KothariMSSW19}), which again rely on the assumption that $\realD$ is a product distribution. 

Three recent lines of work address the~\cite{BriestCKW15,hart2019selling} constructions in a different manner. First,~\citet{ChawlaTT19,ChawlaTT20} consider the related \emph{buy-many} model (where the auctioneer cannot prevent the buyer from interacting multiple times with the auction). \citet{ChawlaTT19} establishes that selling separately achieves an $O(\log k)$-approximation to the optimal buy-many mechanism in \emph{quite} general settings (including the settings considered in this paper). In a different direction,~\citet{PsomasSW19} uses the lens of smoothed analysis~(\citet{SpielmanT04}) to reason about the robustness of the~\cite{hart2019selling} constructions. Finally, \citet{carroll2017robustness} considers a correlation robust framework in which the valuation profile is drawn from a correlated distribution that is not completely known to the seller; the goal is to design a mechanism that maximizes the worst-case (over correlations) seller revenue, when only the items' marginal distributions are known. \cite{carroll2017robustness} shows that selling each item separately is optimal; see \citet{bei2019correlation,gravin2018separation} for further work in this model.

%% file: converse.tex
\section{A Converse to Theorem~\ref{thm:newHN}: The~\cite{hart2019selling} Framework is WLOG}\label{sec:converse}
In this section, we prove Theorem~\ref{thm:main}. Our proof has two main parts. First, we will take the optimal auction for $\realD$ (or one that is arbitrarily close to optimal) and repeatedly simplify it through a sequence of lemmas, at the cost of small constant-factors of revenue. The second part takes this simple menu and draws a connection to MenuGap.

\subsection{Simplifying the Optimal Mechanism}
\label{sec:simple}

We show that for every $\realD$, an approximately-optimal mechanism exists satisfying some useful properties. Our first step simply argues that we may ignore menu options with low prices.

\begin{definition} We say that a mechanism $M$ is \emph{$c$-expensive} if every option has price at least $c$.
\end{definition}

\begin{claim} 
\label{claim:price}
For all $c \in \mathbb{R}_{\geq 0}$, all distributions $\realD$, and all mechanisms $M$, there exists a $c$-expensive mechanism $M'$ satisfying $\rev{\realD, M'} \geq \rev{\realD,M}-c$. 
\end{claim}

\begin{proof}
Take $M'$ to be exactly the same as $M$, except having removed all entries with price $<c$. For every value in the support of $\realD$ with $p^M(\vec{v}) \geq c$ in $M$, we still have $p^{M'}(\vec{v}) \geq c$. This is simply observing that $\vec{v}$'s favorite option in $M$ is still available in $M'$, and all options in $M'$ were also available in $M$. For any value with $p^M(\vec{v}) < c$, we clearly have $p^{M'}(\vec{v}) \geq 0$. So for all $\vec{v}$, we have $p^{M'}(\vec{v}) \geq p^M(\vec{v}) - c$, and the claim follows by taking an expectation with respect to $\vec{v}$.
\end{proof} 

Our next step will show that we can assume further structure on the prices charged, at the cost of a factor of two.

\begin{definition}
A $c$-expensive mechanism $M$ is \emph{oddly-priced} (respectively, \emph{evenly-priced}) if for all $\vec{v}$, there exists an odd (respectively, even) integer $i$ such that $p^M(\vec{v}) \in [c \cdot 2^i, c \cdot 2^{i+1})$. 
\end{definition}

\begin{claim}\label{claim:buckets}
For all $c$-expensive mechanisms $M$, and all $\realD$, there exists either an oddly-priced or evenly-priced $c$-expensive mechanism $M'$ satisfying $\rev{\realD,M'} \geq \rev{\realD,M}/2$. 
\end{claim}
\begin{proof}
Simply let $M_1$ denote the set of menu options from $M$ whose price lies in $[c\cdot 2^i,c\cdot 2^{i+1})$ for an odd integer $i$, and $M_2$ denote the remaining menu options (which lie in $[c \cdot 2^i,c \cdot 2^{i+1})$ for an even power of $i$). It is easy to see that $M_1$ is oddly-priced and $M_2$ is evenly-priced. Then for all $\vec{v}$, we must have $p^{M_1}(\vec{v}) + p^{M_2}(\vec{v}) \geq p^M(\vec{v})$. This is because $\vec{v}$'s favorite menu option from $M$ appears in one of the two menus, and is necessarily $\vec{v}$'s favorite option on that menu (and they pay non-zero from the other menu). Taking an expectation with respect to $\vec{v}$ yields that $\rev{\realD,M_1} +\rev{\realD,M_2} \geq \rev{\realD,M}$, completing the proof.
\end{proof}

This concludes our simplification of the mechanism. In the subsequent section, we draw a connection between MenuGap and the revenue of oddly- or evenly-price $c$-expensive mechanisms.

\subsection{Connecting Structured Mechanisms to MenuGap}
\label{sec:connect}

We begin with the following definition, which describes our proposed $X,Q$ based on a structured mechanism for $\realD$.

\begin{definition}[Representative Sequences] Let $M$ be a $c$-expensive mechanism which is oddly-priced or evenly-priced, and let $\realD$ be any distribution. An $\varepsilon$-\emph{representative sequence} for $M, \realD$ is the following:
\begin{itemize}
    \item Define $\vec{q}_0(\realD,M):=(0,\ldots, 0)$.
    \item Define offset $a$ to be $1$ if $M$ is oddly-priced, and $0$ if $M$ is evenly-priced.
    \item For all $j \in \mathbb{N}_+$, define $B_j:=\{\vec{v} \in \text{supp}(\realD) : p^M(\vec{v}) \in [c\cdot 2^{2(j-1)+a},c\cdot 2^{2(j-1)+a+1})\}$. 
    \item For all $j \in \mathbb{N}_+$, let $\vec{x}_j\in B_j$ be such that $||\vec{x}_j||_1 \leq ||\vec{v}||_1\cdot (1+\varepsilon)$ for all $\vec{v} \in B_j$.\footnote{Note that for all $\varepsilon>0$, such an $\vec{x}_j$ exists, even if $B_j$ is not closed (as long as $B_j$ is non-empty). In particular, because $M$ is $c$-expensive, we know that $||\vec{v}||_1 \geq c > 0$ for all $\vec{v}$ who pay $\geq c$. If $B_j$ is empty, instead omit $\vec{x}_j,\vec{q}_j$ from both lists (i.e. decrease all future indices by one).}
    \item For all $j \in \mathbb{N}_+$, let $\vec{q}_j:=\vec{q}^M(\vec{x}_j)$.
\end{itemize}
\end{definition}

\begin{proposition}\label{prop:main} Let $M$ be a $c$-expensive, and oddly- or evenly-priced. Let $(X,Q)$ be an $\varepsilon$-representative sequence for $M,\realD$. Then: $$\MenuGap{X,Q}\geq \frac{\rev{\realD,M}}{4(1+\varepsilon)\brev{\realD}}\qquad.$$
\end{proposition}

\begin{proof}
The proof will follow immediately from two technical claims. The first claim relates $||\vec{x}_i||_1$ and $\brev{\realD}$.

\begin{claim} 
\label{claim:ell1}
$(1+\varepsilon)\cdot \brev{\realD} \geq \norm{\vec{x}_i}_1 \cdot  \Pr_{\vec{v} \sim \realD}[\vec{v} \in B_i]$.
\end{claim}

\begin{proof}
Recall that $(1+\varepsilon)\cdot ||\vec{v}||_1 \geq ||\vec{x}_i||_1$ for all $\vec{v} \in B_i$. Therefore, if we set a price of $\norm{\vec{x}_i}_1/(1+\varepsilon)$ for the grand bundle, every $\vec{v} \in B_i$ would choose to purchase the grand bundle. This immediately implies the claim, as: \[ \brev{\realD} \geq \frac{||\vec{x}_i||_1}{1+\varepsilon}\cdot \Pr_{\vec{v} \sim \realD}[||\vec{v}||_1 \geq ||\vec{x}_i||_1/(1+\varepsilon)] \geq \frac{||\vec{x}_i||_1}{1+\varepsilon} \cdot \Pr_{\vec{v} \sim \realD}[\vec{v} \in B_i]. \qedhere \]
\end{proof} 

The second claim relates $\rev{\realD,M}$ and $\MenuGap{X,Q}$. Crucially, this claim uses the fact that the mechanism is either oddly-priced or evenly-priced (and therefore $p^M(\vec{x}_i), p^M(\vec{x}_j)$ differ by at least a factor of $2$, for any $i \neq j$).

\begin{claim}\label{claim:menugap} If $X,Q$ is an $\varepsilon$-representative sequence for $M,\realD$, then $\gap_i^{X,Q} \geq p^M(\vec{x}_i)/2$.
 \end{claim}
\begin{proof}
Recall that $\gap_i^{X,Q}:= \min_{j < i}\{\vec{x}_i \cdot (\vec{q}_i - \vec{q}_j)\}$, and that $\vec{q}_i:=\vec{q}^M(\vec{x}_i)$. For any fixed $j < i$, recall that because $M$ was a truthful mechanism, we must have:

\begin{align*}
    &\vec{x}_i \cdot \vec{q}^M(\vec{x}_i) - p^M(\vec{x}_i) \geq \vec{x}_i \cdot \vec{q}^M(\vec{x}_j) - p^M(\vec{x}_j)\\
    \Rightarrow &\vec{x}_i \cdot (\vec{q}_i - \vec{q}_j) \geq p^M(\vec{x}_i)- p^M(\vec{x}_j)\\
    \Rightarrow &\vec{x}_i \cdot (\vec{q}_i - \vec{q}_j) \geq p^M(\vec{x}_i)/2.
\end{align*}
The first line is simply restating incentive compatibility. The second line is basic algebra, and substituting $\vec{q}_i:=\vec{q}^M(\vec{x}_i)$. The third line invokes the fact that $p^M(\vec{x}_i) \geq 2^{2(i-1)+a}$, while $p^M(\vec{x}_j) < 2^{2(j-1)+a+1} \leq 2^{2(i-1)+a-1}$.

\end{proof}

With these two claims, we can wrap up the proof of Proposition~\ref{prop:main}. We can write the following:

\begin{align*}
    \MenuGap{X,Q}&:= \sum_{i=1}^\infty \frac{\gap_i^{X,Q}}{||\vec{x}_i||_1}\\
    &\geq \sum_{i=1}^\infty \frac{p^M(\vec{x}_i)}{2\cdot ||\vec{x}_i||_1}\\
    &\geq \sum_{i=1}^\infty \frac{p^M(\vec{x}_i)\cdot \Pr_{\vec{v} \sim \realD}[\vec{v} \in B_i]}{2 \cdot (1+\varepsilon)\brev{\realD}}\\
    &\geq \sum_{i=1}^\infty \frac{\mathbb{E}_{\vec{v} \sim \realD|\vec{v} \in B_i}[p^M(\vec{v})] \cdot \Pr_{\vec{v} \sim \realD}[\vec{v} \in B_i]}{4 \cdot (1+\varepsilon)\brev{\realD}}\\
    &\geq \frac{\rev{\realD,M}}{4\cdot (1+\varepsilon)\brev{\realD}}
\end{align*}

Above, the first line is simply the definition of MenuGap. The second line uses Claim~\ref{claim:menugap}. The third line uses Claim~\ref{claim:ell1}. The fourth line uses the fact that for all $\vec{v} \in B_i$, $p^M(\vec{v}) \leq 2^{2(i-1)+a+1} = 2 \cdot 2^{2(i-1)+a} \leq 2p^M(\vec{x}_i)$. The final line is just rewriting the definition of $\rev{\realD,M}$, and using the fact that every $\vec{v} \in \text{supp}(\realD)$ with $p(\vec{v}) > 0$ is in some $B_i$.
\end{proof}

\subsection{Wrapping Up}
Proposition~\ref{prop:main}, together with the fact that an oddly- or evenly-priced mechanism is guaranteed to get a good approximation to the optimum, now suffices to prove Theorem~\ref{thm:main}.

\begin{proof}[Proof of Theorem~\ref{thm:main}]
The proof will be a simple consequence of the technical lemmas in this section, once we set $c$ and $\varepsilon$ appropriately. In particular, set $c=\rev{\realD}/100$, and $\varepsilon = 1/100$. Note that $\varepsilon$-representative sequences are guaranteed to exist for any $M, \realD$, as $\varepsilon > 0$.

Start by letting $M'$ denote the $c$-expensive mechanism promised by Claim~\ref{claim:price}. Let then $M''$ denote the oddly- or evenly-priced mechanism promised by Claim~\ref{claim:buckets}. Finally, let $X,Q$ denote the $\varepsilon$-representative sequence for $M'',\realD$. We get:
\begin{align*}
   \MenuGap{X,Q} &\geq^{\text{(Proposition~\ref{prop:main}})} \frac{\rev{\realD,M''}}{{4(1+\varepsilon)\brev{\realD}}}\\
    &\geq^{\text{(Claim~\ref{claim:buckets})}} \frac{\rev{\realD,M'}}{8(1+\varepsilon)\brev{\realD}}\\
    &\geq^{\text{(Claim~\ref{claim:price})}} \frac{\rev{\realD}-c}{8(1+\varepsilon)\brev{\realD}}\\
    &\geq \frac{\rev{\realD}}{9\brev{\realD}}. \qedhere
\end{align*}
%The first line follows immediately from Proposition~\ref{prop:main}. The second follows immediately from Claim~\ref{claim:buckets}. The third is immediately from Claim~\ref{claim:price}. The final line follows by substituting in our choice of $c$ and $\varepsilon$. 
\end{proof}

%\begin{proof}[Proof of Theorem~\ref{thm:ext}]

%% file: separation.tex
\section{No Converse to Lemma~\ref{lem:align}: Separating MenuGap and AlignGap}
\label{sec:sep}

In this section we prove our second main result: Lemma~\ref{lem:align} does \emph{not} admit a converse, even approximately. We briefly remind the reader that \emph{all} previous constructions witnessing $\rev{\realD}/\brev{\realD} = \infty$ arose by establishing sequences $X$ with $\ScaGap{X}=\infty$ (in fact, even $\SupGap{X} = \infty$). Theorem~\ref{thm:main2} establishes that constructions exist outside of this more restrictive framework. We now proceed with the proof of Theorem~\ref{thm:main2}, beginning with a description of our sequence $X$.

\subsection{Description of our construction}\label{sec:description}
We describe our infinite sequence $X$, which consists of consecutive \emph{layers} of points, and note that this aspect of the construction is similar to that in~\cite{HNv1}. For $\ell = 2$ to $\infty$, layer $\ell$ will have $n_\ell:=\ell \lceil \ln^2(\ell) \rceil +1$ points. These points/vectors will have $\ell_2$ norm equal to one, and will be evenly spaced (in terms of their angle) between $(1,0)$ and $(0,1)$. If $\ell$ is even, they will go counterclockwise from $(1,0)$ to $(0,1)$. If $\ell$ is odd, they will go clockwise from $(0,1)$ to $(1,0)$. Specifically: 
\begin{itemize}
\item Define $n_\ell:=\ell \lceil \ln^2(\ell) \rceil +1$. Define $\theta_\ell:=\frac{\pi}{2(n_\ell-1)}$.
\item Point $\vec{x}_{\ell,j}$ is the $j^{th}$ point in the $\ell^{th}$ layer, and is $(\cos(j\theta_\ell),\sin(j\theta_\ell) )$ when $\ell$ is even, or $(\sin(j\theta_\ell),\cos(j\theta_\ell))$ when $\ell$ is odd.
\item The infinite sequence $X$ is $((\vec{x}_{\ell,j})_{j=0}^{n_\ell-1})_{\ell=2}^{\infty}$. 
\end{itemize}

\noindent Figure~\ref{fig:samplelayer} demonstrates two layers of our construction. In the remainder of this section, we may refer to the sequence of points $\vec{x}_{\ell, j}$ by a single indexed sequence $\vec{x}_i$. The latter is the same as the former where points are ordered lexicographically with respect to the original indexing.

\begin{figure}
\centering
  \includegraphics[width=.6\linewidth]{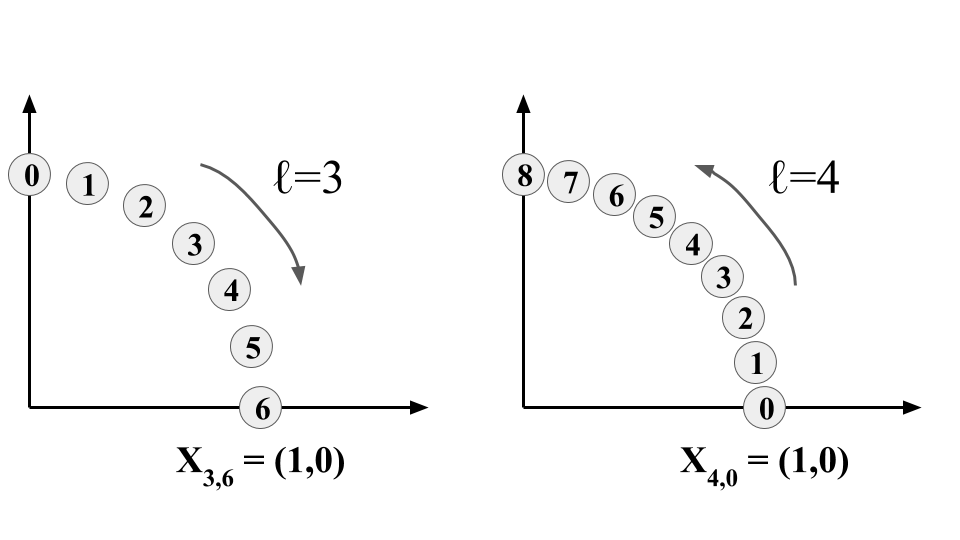}
  \captionof{figure}{An illustration of two layers of our construction. The number of points in each layer increases with $\ell$, but they are always evenly spaced between $(1,0)$ and $(0,1)$. The direction in which the points are placed alternates between clockwise and counterclockwise.}
  \label{fig:samplelayer}
  \end{figure}

\subsection{Upper Bounding $\ScaGap{X}$ via Lagrangian duality}

Now that we have our construction, we first need to upper bound $\ScaGap{X}$ and establish that it's finite. To this end, first observe that for any sequence $X$, $\ScaGap{X}$ is the solution to the following (infinite, if $X$ is infinite) mathematical program, where the variables are $\sgap_i, c_i$ (the sequence $\vec{x}_i$ is fixed, as we're aiming to compute $\ScaGap{X}$):
\begin{align*}
\ScaGap{X} :&\begin{cases}
&\max  \sum_i \max\{0,\sgap_i\}/||\vec{x}_i||_1  \\ 
\text{subject to: } &\forall i, j < i: \sgap_i \leq \vec{{x}}_i \cdot (c_i \vec{{x}}_i - c_j \vec{{x}}_j) \\
&\forall i: 0 \leq c_i \leq 1/||\vec{x}_i||_\infty
\end{cases}
\end{align*}

We next proceed with a series of relaxations of this program. Some steps are specific to our choice of $X$ from Section~\ref{sec:description}, while others hold for arbitrary $X$. Our first step is specific to this construction, and simply bounds $||\vec{x}||_1$. Consider the following mathematical program:

\begin{align*}
\mathrm{AlignGap'}(X) :&\begin{cases}
&\max  \sum_i \max\{0,\sgap_i\}  \\ 
\text{subject to: } &\forall i, j < i: \sgap_i \leq \vec{{x}}_i \cdot (c_i \vec{{x}}_i - c_j \vec{{x}}_j) \\
&\forall i: 0 \leq c_i \leq \sqrt{2}.
\end{cases}
\end{align*}

\begin{observation}\label{obs:prime} For the sequence $X$ defined in Section~\ref{sec:description}, $\mathrm{{AlignGap}}'(X) \geq \ScaGap{X}$.
\end{observation}
\begin{proof}
Every $\vec{x}_i$ in the construction has $||\vec{x}_i||_1 \geq ||\vec{x}_i||_2 = 1$. Therefore, the new objective function is only larger. Moreover, every $\vec{x}_i$ in the construction has $||\vec{x}_i||_\infty \geq 1/\sqrt{2}$ (because the $\ell_2$ norm is $1$), so this is relaxing the upper bound on $c_i$.
\end{proof}

We will proceed to upper bound $\text{AlignGap}'(X)$ via a Lagrangian relaxation of the formulation above. Specifically, consider the following Lagrangian relaxation. We put a Lagrangian multiplier of $1$ on every constraint of the form $\sgap_i \leq \vec{x}_i \cdot (c_i \vec{x}_i - c_{i-1}\vec{x}_{i-1})$, for all $i > 1$. We put a Lagrangian multiplier of $0$ on all other constraints involving $\sgap$. We will \emph{not} put a Lagrangian multiplier on constraints binding $c_i$ to $[0,\sqrt{2}]$, and keep those in the program. This yields the following Lagrangian relaxation (for simplicity of notation below, define $c_0 :=0$, and define $\vec{x}_0 = \vec{0}$):

\begin{align*}
\mathrm{\text{LagRel}_1(X)} :&\begin{cases}
&\max  \sum_i \max\{0,\sgap_i\} + \vec{x}_i \cdot (c_i \vec{x}_i - c_{i-1} \vec{x}_{i-1}) - \sgap_i\\ 
\text{subject to: } 
&\forall i: 0 \leq c_i \leq \sqrt{2}
\end{cases}
\end{align*}

\begin{observation}\label{obs:lagrange} For all $X$, $\mathrm{AlignGap}'(X) \leq \LagRelOne{X}$.
\end{observation}
\begin{proof}
This follows immediately from weak Lagrangian duality. For a quick refresher on weak Lagrangian duality, observe that for any feasible solution to the LP defining $\mathrm{AlignGap}'(X)$ we must have $\vec{x}_i \cdot (c_i \vec{x}_i - c_{i-1}\vec{x}_{i-1}) - \sgap_i \geq 0$. Therefore, for any feasible solution to the original LP, that solution is also feasible for $\mathrm{\text{LagRel}}_1(X)$, and the objective is only larger. Therefore, the optimal solution to $\mathrm{\text{LagRel}}_1(X)$ must be at least as large as $\mathrm{AlignGap}'(X)$.
\end{proof}

We now proceed to further simplify $\LagRelOne{X}$. The next step is defined below:

\begin{align*}
\mathrm{\text{LagRel}_2(X)} :&\begin{cases}
&\max  \sum_i  \vec{x}_i \cdot (c_i \vec{x}_i - c_{i-1} \vec{x}_{i-1})\\ 
\text{subject to: } 
&\forall i: 0 \leq c_i \leq \sqrt{2}
\end{cases}
\end{align*}

\begin{observation}\label{obs:lagrange2} For all $X$, $\LagRelOne{X} = \LagRelTwo{X}$.
\end{observation}
\begin{proof}
Observe that for all $i$, $\max\{0,\sgap_i\} - \sgap_i \leq 0$. When $\sgap_i = 0$, the maximum is achieved (and $\sgap_i:=0$ is feasible). Substituting $\max\{0,\sgap_i\} - \sgap_i = 0$ for all $i$ concludes the proof.
\end{proof}

We make one last observation about the relaxation, which simply rewrites the objective function to group all coefficients of $c_i$. For ease of notation below, define $\vec{x}_{N+1}:=\vec{0}$ (if $N$ is finite. If $N = \infty$, there are no notational issues).

\begin{align*}
\mathrm{\text{LagRel}(X)} :&\begin{cases}
&\max  \sum_i c_i \vec{x}_i \cdot (\vec{x}_i -\vec{x}_{i+1})\\ 
\text{subject to: } 
&\forall i: 0 \leq c_i \leq \sqrt{2}
\end{cases}
\end{align*}

\begin{observation}\label{obs:Lagrange3} For all $X$, $\mathrm{AlignGap}'(X) \leq \LagRelOne{X} = \LagRelTwo{X} = \LagRel{X}$.
\end{observation}

Now, we move to analyze $\LagRel{X}$ for our particular sequence $X$. 

\begin{claim}\label{claim:X} For the sequence $X$ defined in Section~\ref{sec:description}, $\LagRel{X}=\sqrt{2} \cdot \sum_i 1 - \vec{x}_i \cdot \vec{x}_{i+1}$.
\end{claim}
\begin{proof}
Observe here that $\vec{x}_i \cdot (\vec{x}_i-\vec{x}_{i+1}) \geq 0$ for all $i$, as each $\vec{x}_i$ has $\ell_2$ norm exactly one. This means that the optimal solution for $\LagRel{X}$ sets each $c_i:=\sqrt{2}$. Finally, recalling that $\vec{x}_i \cdot \vec{x}_i = 1$ for all $i$ concludes the claim.
\end{proof}

\begin{proposition}\label{prop:mainduality} For the sequence $X$ defined in Section~\ref{sec:description}, $\LagRel{X} \leq 6$.
\end{proposition}
\begin{proof}
Let us first observe that if $\vec{x}_i$ is the last point in a layer, then in fact $\vec{x}_{i+1} = \vec{x}_i$, and therefore $1-\vec{x}_i \cdot \vec{x}_{i+1} = 0$. Therefore, these terms do not contribute to the sum. We can then rewrite the term to sum over all layers as follows:
\begin{align*}
\LagRel{X}&=\sqrt{2}\cdot \sum_{\ell}\sum_{j=0}^{n_\ell-2} 1-\vec{x}_{\ell,j}\cdot \vec{x}_{\ell,j+1}\\
&=\sqrt{2} \cdot \sum_{\ell}(n_\ell-1) \cdot (1-\cos(\theta_\ell))\\
&\leq \sqrt{2} \cdot \sum_\ell (n_\ell-1) \cdot \theta_\ell^2/2\\
&= \sqrt{2} \cdot \sum_\ell \frac{\pi^2}{8 (n_\ell-1)}\\
&=\sqrt{2} \cdot \sum_\ell \frac{\pi^2}{8 \ell \lceil \ln^2(\ell) \rceil}\\
&\leq 6.
\end{align*}

Above, the first line follows by the reasoning in the first paragraph. The second line follows by observing that the angle between any two points in layer $\ell$ is exactly $\theta_\ell$ (and the two points in question of $\ell_2$ norm equal to one). The third line follows as $\cos(\theta_\ell) \geq 1-\theta_\ell^2/2$ for any $\theta_\ell \in [0,\pi/2]$ (and all $\theta_\ell$ are indeed in $[0,\pi/2]$). The fourth line follows by substituting the definition of $\theta_\ell$ as a function of $n_\ell$. The fifth line follows by definition of $n_\ell$. The final line is just calculation for this particular infinite series.\footnote{This follows as $\sum_{\ell=2}^\infty \frac{1}{\ell \ln^2(\ell)} \leq 3$.}
\end{proof}

Observation~\ref{obs:prime}, Observation~\ref{obs:Lagrange3}, and Proposition~\ref{prop:mainduality} yield the main result of this section:

\begin{proposition}\label{prop:scagap} For the sequence $X$ defined in Section~\ref{sec:description}, $\ScaGap{X} \leq 6$.
\end{proposition}

\subsection{Step 3: Picking a $Q$ to Lower Bound $\MenuGap{X}$.}
Finally, we propose a sequence $Q$ and show that $\MenuGap{X,Q} = \infty$. We describe the sequence again in layers, to match our description of $X$ (that is, the vector $\vec{q}_{\ell,j}$ corresponds to the vector $\vec{x}_{\ell,j}$). In particular, for each even layer $\ell$, the vectors $\vec{q}_{\ell,j}$ will have a fixed $x$-coordinate, and the $y$-coordinate will increase with $j$. For each odd layer, we will introduce no new vectors (i.e. we will just let $\gap_i^{X,Q} = 0$ for all $i$ in an odd layer). Specifically, the construction is as follows:
\begin{itemize}
\item Define $\alpha:= \sum_{\ell=2}^\infty \frac{1}{\ell \ln^2(\ell)}$ (and note that $\alpha < \infty$).
\item Define $z_\ell:= \frac{1}{\alpha} \sum_{j=2}^\ell \frac{1}{j\ln^2(j)}$. 
\item Define $\delta_\ell:= z_\ell-z_{\ell-1} = \frac{1}{\alpha \ell \ln^2(\ell)}$.
\item For $j < n_\ell-1$, define $z_{\ell,j}:=1-\delta_\ell \cot((j+1)\theta_\ell)$. For $j = n_\ell-1$, define $z_{\ell,j}:=1$.\footnote{To see that $z_{\ell,j} \geq 0$ observe that $\cot(x) \leq 1/x$ for all $x \in [0, \pi/2]$. Therefore we get $\delta_\ell \cot((j+1) \theta_\ell) \leq \delta_\ell \frac{1}{(j+1) \theta_\ell} \leq \delta_\ell \frac{2 (n_\ell -1) }{\pi (j+1)} \leq \frac{2}{\pi} \delta_\ell n_\ell = \frac{2 \ell \lceil \log^2 \ell \rceil}{\pi \alpha \ell \log^2 \ell}\leq 1$, for all $\ell, j$ ($\alpha > 1.9$). To see that $z_{\ell, j} \leq 1$ simply note that the term we subtract can't be negative.}
\item For all even $\ell$, and all $j$, set $\vec{q}_{\ell,j}:= (z_\ell, z_{\ell,j})$.
\item For all odd $\ell$, and all $j$, set $\vec{q}_{\ell,j}:=\arg\max_{\ell' < \ell,j'}\{\vec{x}_{\ell, j} \cdot\vec{q}_{\ell',j'}\}$. 
\end{itemize}

\begin{proposition}
\label{prop:vecgapdiverges}
$\MenuGap{X,Q}=\infty$. 
\end{proposition}

\begin{proof}
To ease notation throughout the proof, we'll use the notation $\gap_{\ell,j}^{X,Q}:=\gap_i^{X,Q}$, where $\vec{x}_i:=\vec{x}_{\ell,j}$ ($\vec{x}_i$ is the $j^{th}$ point on layer $\ell$). We will also use the notation $(\ell',j') < (\ell,j)$ if $\ell' < \ell$, or $\ell' = \ell$ and $j' < j$ (that is, if the $j'^{th}$ point in the $\ell'^{th}$ layer comes before the $j^{th}$ point in the $\ell^{th}$ layer). To understand $\gap_{\ell,j}^{X,Q}$, we need to understand which point ``sets the gap'' for $\vec{x}_{\ell,j}$, that is, which $(\ell',j'):=\arg\min_{(\ell',j')< (\ell,j)}\{(\vec{q}_{\ell,j}-\vec{q}_{\ell',j'})\cdot \vec{x}_i\}$.

We first analyze which point sets the gap for $\vec{{x}}_{\ell,j}$ (for even $\ell$; for odd $\ell$ the gap is zero and we don't care which point sets it), and observe that it must either be $\vec{q}_{\ell,j-1}$ or $\vec{q}_{\ell-2,n_{\ell-2}-1}$ (that is, it must be the previous point in the same layer, or the final point in the previous even layer). 

\begin{claim}\label{claim:gap}
For all $j$, and all even $\ell$, $\gap_{\ell,j}^{X,Q}= \vec{x}_{\ell,j}\cdot \vec{q}_{\ell,j} - \max\{\vec{x}_{\ell,j} \cdot \vec{q}_{\ell-2,n_{\ell-2}-1}, \vec{x}_{\ell,j} \cdot \vec{q}_{\ell,j-1}\}$.\footnote{For simplicity of notation, define $\vec{q}_{0,j} = \vec{0} = \vec{q}_{\ell,-1}$ for all $\ell, j$.}
\end{claim}
\begin{proof}
First, note that $\gap_{\ell,j}^{X,Q}:=\min_{(\ell',j') < (\ell,j)} \{\vec{x}_{\ell, j} \cdot (\vec{q}_{\ell, j} - \vec{q}_{\ell',j'})\} = \vec{x}_{\ell, j} \cdot \vec{q}_{\ell, j} - \max_{(\ell',j')<(\ell,j)}\{\vec{x}_{\ell, j}\cdot \vec{q}_{\ell',j'}\}$. Observe that the first component of $\vec{q}_{\ell',j'}$ %($ \frac{1}{\alpha} \sum_{j=2}^{\ell'} \frac{1}{j\ln^2(j)}$ for even and the argmax over $z < \ell'$ for odd $\ell'$) 
is monotone increasing in $\ell'$ (for fixed $j'$), and the second component %($1-\delta_{\ell'} \cot((j'+1)\theta_{\ell'})$ for even $\ell'$ and $j' < n_{\ell'}-1$, and an argmax over $z < \ell'$ for odd $\ell'$)  
is monotone increasing in $j'$ (for fixed $\ell'$). Moreover, the second component of $\vec{q}_{\ell', n_{\ell'}-1}$ is $1$, and this is the maximum possible. Also, both components of $\vec{x}_{\ell,j}$ are non-negative, and therefore we conclude that $\vec{x}_{\ell,j} \cdot \vec{q}_{\ell-2,n_{\ell-2}-1} \geq \vec{x}_{\ell,j} \cdot \vec{q}_{\ell',j'}$ whenever $(\ell',j') \leq (\ell-2,n_{\ell-2}-1)$ (in fact, this extends even to $(\ell',j') \leq (\ell-1,n_{\ell-1}-1)$ as no new $\vec{q}$ are introduced in layer $\ell-1$). Also, $\vec{x}_{\ell,j} \cdot \vec{q}_{\ell,j-1} \geq \vec{x}_{\ell,j} \cdot \vec{q}_{\ell,j'}$ whenever $j' \leq j-1$.
\end{proof}

Now that we know that the gap is set either by the last point in the previous layer, or the previous point in the current layer, we can nail down $\gap_{\ell,j}^{X,Q}$ exactly.

\begin{lemma}\label{lem:gap} For all even $\ell > 2$, and all $j \in [0,n_\ell-1]$:
\[ \gap^{X, Q}_{\ell,j} \geq \delta_\ell \frac{\sin(\theta_\ell)}{\sin((j+1)\theta_\ell)}. \]
\end{lemma}
\begin{proof}
To prove the lemma, we simply compute the inner product of $\vec{x}_{\ell,j}$ with the three relevant vectors $\vec{q}_{\ell,j}, \vec{q}_{\ell-2,n_{\ell-2}-1}, \vec{q}_{\ell,j-1}$. To this end, recall that:
\begin{align*}
\vec{q}_{\ell,j} &= \left(z_\ell, 1 - \delta_\ell \cot((j+1) \theta_\ell)\right), \\ 
\vec{q}_{\ell,j-1}&= \left(z_\ell, 1 - \delta_\ell \cot(j\theta_\ell)\right), \\
\vec{q}_{\ell-2,n_{\ell-2}-1}&= \left(z_{\ell-2}, 1 \right).
\end{align*}

Therefore, observe that  
\begin{align*}
\vec{x}_{\ell,j} \cdot (\vec{q}_{\ell,j} - \vec{q}_{\ell,j-1}) &= \sin(j\theta_\ell)\cdot\delta_\ell\cdot  \left(\cot(j\theta_\ell) - \cot((j+1)\theta_\ell) \right) \\
&= \sin(j \theta_\ell) \cdot \delta_\ell\cdot \left(\frac{\cos(j\theta_\ell)}{\sin(j\theta_\ell)} - \frac{\cos((j+1)\theta_\ell)}{\sin((j+1)\theta_\ell)} \right) \\ 
&= \delta_\ell \cdot  \frac{\cos(j\theta_\ell)\sin( (j+1)\theta_\ell) - \sin(j\theta_\ell)\cos( (j+1)\theta_\ell)}{\sin((j+1)\theta_\ell)}\\
&=\delta_\ell \cdot \frac{\sin(\theta_\ell)}{\sin( (j+1)\theta_\ell)}.
\end{align*}

Similarly, 
\begin{align*}
\vec{x}_{\ell,j} \cdot (\vec{q}_{\ell,j} - \vec{q}_{\ell-2,n_{\ell-2}-1}) &= (\delta_\ell+\delta_{\ell-1}) \cdot \cos(j \theta_\ell) - \delta_\ell\cot( (j+1)\theta_\ell) \cdot \sin(j\theta_\ell) \\
&\geq \delta_\ell \cdot \cos(j \theta_\ell) - \delta_\ell\cot( (j+1)\theta_\ell) \cdot \sin(j\theta_\ell) \\
&=\frac{\delta_\ell}{\sin((j+1)\theta_\ell)}\left(\sin((j+1)\theta_\ell)\cos(j\theta_\ell) - \sin(j \theta_\ell) \cos((j+1)\theta_\ell) \right) \\
&= \delta_\ell \frac{\sin(\theta_\ell)}{\sin((j+1)\theta_\ell)}.
\end{align*}

This means that no matter which point sets the gap (or if one of the points does not exist), the gap is at least $\delta_\ell \frac{\sin(\theta_\ell)}{\sin((j+1)\theta_\ell)}$.
\end{proof}

Finally, now that we know the gap for each point on an even layer, we just need to sum over each even layer.

\begin{corollary}\label{cor:gap} For any even $\ell > 2$, $\sum_{j=0}^{n_\ell-1} \gap_{\ell,j}^{X,Q} \geq \delta_\ell \cdot \ln(n_\ell)/2$.
\end{corollary}

\begin{proof}
Consider the following sequence of calculations:

\begin{align*}
\sum_{j=0}^{n_\ell-1} \gap_{\ell,j}^{X,Q} &\geq \sum_{j=0}^{n_\ell-1} \delta_\ell \frac{\sin(\theta_\ell)}{\sin((j+1)\theta_\ell)}\\
&\geq \delta_\ell \cdot (\theta_\ell-\theta_\ell^3/6)\cdot \sum_{j=0}^{n_\ell-1} \frac{1}{(j+1)\theta_\ell}\\
&\geq \delta_\ell \cdot (1-\theta_\ell^2/6) \cdot \ln(n_\ell)\\
&\geq \delta_\ell \cdot \ln(n_\ell)/2
\end{align*}

Above, the first line follows from Lemma~\ref{lem:gap}. The second line uses the fact that $\theta_\ell - \theta_\ell^3/6 \leq \sin(\theta_\ell) \leq \theta_\ell$, because $\theta_\ell \in [0,\pi/2]$. The third line follows as the $n$-th harmonic sum is at least $\ln(n)$. The final line follows as $\theta_\ell^2/6 = \pi^2/(24 (n_\ell-1)^2) \leq 1/2$.
\end{proof}

And finally, we can wrap up the proof of the proposition. Here, we just need to recall that $\delta_\ell:=\frac{1}{\alpha n_\ell} = \frac{1}{\alpha \ell \ln^2(\ell)}$. Therefore, we conclude that:
\begin{align*}
\sum_{\ell \text{ even}} \sum_{j=0}^{n_\ell-1} \gap_{\ell,j}^{X,Q} &\geq \sum_{\ell \text{ even}} \delta_\ell \cdot \ln(n_\ell)/2\\
&=\sum_{\ell \text{ even}} \frac{1}{2 \alpha \ell \ln(\ell)}\\
&= \infty
\end{align*}

Above, the first line simply restates Corollary~\ref{cor:gap}. The second line substitutes the definition of $\delta_\ell$ and $n_\ell$. The final line follows as the indefinite integral of $\frac{1}{x\ln(x)}dx$ is $\ln\ln (x)$ (and $\alpha$ is an absolute constant).
\end{proof}

This completes the proof of Theorem~\ref{thm:main2}: Proposition~\ref{prop:scagap} establishes that $\ScaGap{X} \leq 6$, while Proposition~\ref{prop:vecgapdiverges} establishes that $\MenuGap{X} = \infty$. 

%% file: conclusion.tex
% !TeX root = main.tex
% !TEX root = main.tex

\section{Conclusion}
\label{sec:conclusion}

We study the nature of distributions $\realD$ with $\rev{\realD}/\brev{\realD} = \infty$. Prior work established a framework to construct such distributions, and therefore established sufficient conditions~\cite{BriestCKW15, hart2019selling}. Our first main result establishes that the most general of these frameworks is in fact complete (Theorem~\ref{thm:main}). Our second main result establishes that the more restrictive framework, through which all previous constructions arose, is not complete (Theorem~\ref{thm:main2}). Finally, we build upon our main construction to develop a novel distribution $\realD$ witnessing $\rev{\realD}/\brev{\realD}=\infty$, but for which none of the ``aligned'' mechanisms of prior work can possibly witness this (Corollary~\ref{cor:main}). 

In terms of future work, it remains open as to whether there is an alternative definition for ``sequences with $\MenuGap{X}=\infty$'' which is easier to parse. Our work establishes that understanding such sequences is necessary and sufficient to understand distributions with $\rev{\realD}/\brev{\realD} = \infty$, and establishes that ``sequences with $\MenuGap{X} = \infty$'' is not equivalent to ``sequences with $\ScaGap{X} = \infty$.'' But it would be exciting for future work to better understand sequences with $\MenuGap{X} = \infty$. 

Interestingly, ongoing work by~\cite{AS22} uses the framework developed in this paper in order to prove \emph{approximation results} for so-called fine-grained buy-many mechanisms, a class of mechanisms which interpolates between buy-one and the recently introduced buy-many mechanisms~\cite{ChawlaTT19, ChawlaTT20}. In general, our techniques make it possible for future work to explore the gap between various simple versus optimal benchmarks without having to reason directly about the underlying mechanisms.

%% file: appendix.tex
\section{$\MenuGap{X}=1$ when $k=1$}\label{app:single}

In this brief section we prove that when $k=1$, for any sequence of $x_i \in \mathbb{R}^+_{\geq 0}$, $\MenuGap{X} = 1$. 

\begin{claim}
When $k=1$, for any $X = \{x_i\}_{i=1}^{N}$, $x_i \in \mathbb{R}^+_{\geq 0}$, $\MenuGap{X} = 1$. 
\end{claim}

\begin{proof}
Note that when $k=1$, $||x_i||_1 = x_i$. Therefore, $\MenuGap{X, Q} = \sum_i \min_{j < i} (q_i - q_j)$. We make the following observation which allows us to look at structured optimal solutions. 

\begin{observation}
Any optimal solution $Q$ to $\MenuGap{X}$ is monotone non-decreasing. 
\end{observation} 

\begin{proof}
For the sake of contradiction, suppose we are given an optimal solution $Q$ that is not monotone non-decreasing. Let $i$ be the smallest index for which $q_i < q_{i-1}$. Then $\gap_i^{X,Q} = (q_i - q_{i-1}) x_i < 0$. Consider instead a solution $Q'$ where $q'_j = q_j$ for all $j \neq i$ and $q'_i = q_{i-1}$. Now, $\gap_i^{X,Q'} = 0$. Since $Q_{\leq i-1} = Q'_{\leq i-1}$, $\gap_j^{X,Q} = \gap_j^{X,Q'}$ for all $j < i$. Since $q_{i-1} > q_i$, for any $j >i$ it holds that $(q_j - q_{i-1}) < (q_j - q_i)$. Therefore, $q_i$ is not ``setting the gap'' for any point after it. Hence it also holds that $\gap_j^{X,Q} = \gap_j^{X,Q'}$ for all $j > i$. Putting everything together we get that $\MenuGap{X,Q'} - \MenuGap{X,Q} = \gap_{i}^{X,Q'} - \gap_i^{X,Q} > 0$ contradicting the optimality of $Q$.   
\end{proof}

With this observation in hand, since the $q_i$ are monotone non-decreasing, without loss of generality it holds that $\gap_{i}^{X,Q} = \min_{j < i} q_i - q_{j} = q_i - q_{i-1}$ ($q_{i-1} \geq q_j$ for all $j < i$). Therefore, we get $\MenuGap{X,Q} = \sum_i q_i-q_{i-1} = q_N - q_0$. Since $q_0 = 0$ and $0 \leq q_N \leq 1$, we get that $\MenuGap{X,Q} \leq 1$. 

Finally, note that for any $X$, we can set $q_N = 1$ and $q_i = 0$ for all other $i$, proving that $\MenuGap{X} \geq 1$.

\end{proof}

\section{Omitted Proofs}\label{app:omitted}
\begin{proof}[Proof of Lemma~\ref{lem:align}]
We prove that for all $X,C$, $\ScaGap{X,C} \leq \MenuGap{X}$, which implies the lemma. For a given $X,C$, define:
\begin{itemize}
\item $\vec{q}_i := c_i \cdot \vec{x}_i$, if $\sgap_i^{X,C} > 0$. 
\item $\vec{q}_i:=\arg\max_{j < i}\{c_j \cdot \vec{x}_j\}$, if $\sgap_i^{X,C} \leq 0$. 
\end{itemize}

Observe first that each $\vec{q}_i \in [0,1]^k$, as each $c_i \vec{x}_i \in [0,1]^k$ (this follows because each component of $\vec{x}_i$ is at most $||\vec{x}_i||_\infty$, and each $c_i$ is at most $1/||\vec{x}_i||_\infty$). Next, observe that if $\sgap_i^{X,C} \leq 0$, then $\gap_i^{X,Q} = 0$. This is by definition in bullet two above. Finally, observe that if $\sgap_i^{X,C} > 0$, then $\gap_i^{X,Q} \geq \sgap_i^{X,C}$. This is because the set of $\{\vec{q}_j\}_{j < i}$ is a subset of $\{c_j \vec{x}_j\}_{j < i}$, and because $\vec{q}_i := c_i \cdot \vec{x}_i$ by bullet one. Therefore, $\gap_i^{X,Q} \geq \max\{0,\sgap_i^{X,C}\}$ for all $i$ and the lemma follows. 
\end{proof}

\section{Proof of Corollary~\ref{cor:main}}\label{app:main}
We prove Corollary~\ref{cor:main} by making use of Theorem~\ref{thm:newHN} combined with the sequence $X$ from Section~\ref{sec:description}. The only task is to confirm that $\arev{D} < \infty$ for the resulting $\realD$, which essentially requires that we execute and analyze the construction fully. Let us quickly review the~\cite{hart2019selling} construction, given as input a sequence $X$:
\begin{itemize}
\item Let $B$ be a very large constant, to be defined later.
\item Let $\vec{v}_i:=B^{2^i} \cdot \vec{x}_i/||\vec{x}_i||_1$ (for all $i$).
\item Let $\realD$ sample $\vec{v}_i$ with probability $1/B^{2^i}$ (for all $i$). 
\item Let $\realD$ sample $\vec{0}$ with probability $1-\sum_{i\geq 1}1/B^{2^i}$. 
\end{itemize}

\cite{hart2019selling} establishes that the above construction yields Theorem~\ref{thm:newHN} (for sufficiently large $B$, as a function of $\varepsilon$). To complete the proof of Corollary~\ref{cor:main}, we just need to relate $\arev{D}$ for this construction to $\ScaGap{X}$.

\begin{proposition}\label{prop:hn} The construction above yields a $\realD$ satisfying $\arev{\realD} \leq \ScaGap{X}+1/B$.
\end{proposition}
\begin{proof}
Consider any mechanism $M$. We show that $\ScaGap{X} \geq \arev{\realD,M}-1/B$. To see this, consider the following choice of $C$:
\begin{itemize}
\item If $\vec{v}_i$ is parallel to $\vec{q}^M(\vec{v}_i)$, set $c_i:=||\vec{q}^M(\vec{v}_i)||_2/||\vec{x}_i||_2$. 
\item If $\vec{v}_i$ is not parallel to $\vec{q}^M(\vec{v}_i)$, set $c_i:=0$. 
\end{itemize}

We now need to lower bound $\sgap_i^{X,C}$, when $i$ satisfies the first bullet. Observe that, because $M$ is truthful, we must have, for all $j < i$:

\begin{align*}
\vec{v}_i \cdot \vec{q}^M(\vec{v}_i) - p^M(\vec{v}_i) &\geq \vec{v}_i \cdot \vec{q}^M(\vec{v}_j) - p^M(\vec{v}_j)\\
\Rightarrow p^M(\vec{v}_i) &\leq p^M(\vec{v}_j) + B^{2^i}\vec{x}_i\cdot (c_i \vec{x}_i - c_j \vec{x}_j)/||\vec{x}||_1\\
\Rightarrow p^M(\vec{v}_i) &\leq 2B^{2^{i-1}}+B^{2^i}\sgap_i^{X,C}/||\vec{x}_i||_1
\end{align*}
Above, the first line follows from incentive compatibility. The second line follows as $\vec{q}^M(\vec{v}_i) = c_i \vec{x}_i$ for all $i$ in the first bullet, and either $\vec{q}^M(\vec{v}_j) = c_j \vec{x}_j$, or $c_j = 0$. The final line follows by taking $j:=\arg\min_{j < i}\{\vec{v}_i\cdot (c_i \vec{v}_i - c_j \vec{v}_j)\}$, and by observing that $\vec{v}_j$ cannot possibly pay more than their value for the grand bundle.

We can then conclude that:

\begin{align*}
\arev{\realD,M} &\leq \sum_i (2B^{2^{i-1}}+B^{2^i} \sgap_i^{X,C}/||\vec{x}_i||_1)/B^{2^i}\\
&\leq \sum_i 2/B^{2^{i-1}} + \ScaGap{X}\\
&\leq \ScaGap{X}+1/B.
\end{align*}

\end{proof}

Because we can take $B$ as large as we like, we can construct a $\realD$ such that $\arev{\realD}$ is arbitrarily close to $\ScaGap{X}$, while also maintaining that $\rev{\realD}$ is arbitrarily close to $\MenuGap{X}$. Because Theorem~\ref{thm:main2} provides a construction $X$ such that $\MenuGap{X}/\ScaGap{X} = \infty$, the~\cite{hart2019selling} construction, with sufficiently large $B$, yields a $\realD$ with $\rev{\realD}/\arev{\realD} = \infty$, completing the proof of Corollary~\ref{cor:main}.